\documentclass[12pt]{amsart}
\usepackage{a4wide,palatino,epsfig,latexsym,
amsmath,amsthm,graphicx,hyperref,enumerate,morefloats,color, natbib, soul}

\usepackage{tikz}

\title{Higher order epistasis and fitness peaks}
  \author{Kristina Crona and Mengming Luo}

\theoremstyle{plain}

\newtheorem{theorem}{Theorem}[section]

\theoremstyle{definition}

\newtheorem{definition}[theorem]{Definition}

\newtheorem{example}[theorem]{Example}

\begin{document}

\maketitle

\begin{abstract}
We  show that  higher order epistasis 
has a substantial impact on evolutionary dynamics
by analyzing peaks in the fitness landscapes.
There are 193,270,310 fitness graphs 
for 4-locus systems, i.e., directed acyclic graphs on
4-cubes. The graphs can be partitioned into 511,863 isomorphism classes.
We identify all fitness graphs with 6 or more peaks. 
81 percent of them imply 4-way epistasis, whereas
9 percent of all 4-locus fitness graphs imply 4-way epistasis.
Fitness graphs are useful in that they
reflect the entire collection of fitness landscapes,
rather than particular models.
Our results depend on a characterization of fitness 
graphs that imply $n$-way epistasis. The characterization
is expressed in terms of a partition property that can 
be derived from Hall's marriage theorem for bipartite graphs.
A similar partition condition holds for any partial order.
The result answers an open problem posed at a 
conference on interactions between algebra
 and the sciences at the Max Planck institute.
\end{abstract}

\section{introduction}
Recent empirical studies have detected higher order epistasis
in several  biological systems \citep{sh, wlw}. However, the impact of higher
order epistasis has not been carefully investigated.
The purpose of this study is to relate higher order epistasis
and peaks in fitness landscapes. 
The genotypes for a biallellic 2-locus system are denoted
\[
00, 10, 01, 11.
\]
We define a fitness landscape as
a function $w: \{0,1 \}^n \mapsto \mathbb{R}$,
which assigns a fitness value to
each genotype. The fitness of the genotype $g$ is denoted $w_g$.

Epistasis is measured by the expression
\[
w_{00} +w_{11}-w_{01} -w_{10}.
\]

The genotypes for a biallellic 3-locus system are denoted
\[
000, 100, 010, 001, 110, 101, 011, 111,
\]
and the (total) 3-way epistasis is measured by 
\[
u_{111} = w_{000} + w_{011} + w_{101} + w_{110}  - w_{001} - w_{010} - w_{100} -w_{111}.
\]
We refer to $000, 110, 101, 011$ as even genotypes, and
$100, 010, 001, 111$ as odd genotypes. "Even" and "odd" refer to the 
number of 1's in the genotype labels.
The definition of $n$-way epistasis is analogous.

Our analysis of higher order epistasis and peaks depends on fitness graphs.
The vertices of a fitness graph are labeled by the $2^n$ genotypes in the $n$-locus system.
For each pair of mutational neighbors, i.e., genotypes that differ at one locus only,
an arrow points toward the genotype of higher fitness (Figure 1).
In mathematical terms, a fitness graph is a directed acyclic graph on a hypercube, or a cube orientation.
A peak in a fitness landscape is a genotype of higher fitness than all of its mutational neighbors.
By abuse of notation, vertices in fitness graphs are referred to as peaks if they correspond to peaks
in the fitness landscapes.
Fitness graphs were originally used for analyzing empirical data \citep[e.g.]{dpk, fkd}, and
in more recent time for theoretical results \citep{cgg, wdo,cgb}, 
see also the Discussion section. Fitness graphs are of interest since
they capture important aspects of evolutionary dynamics.

%Kolla "partial order of, relevant for/to% 

\begin{figure}
\begin{tikzpicture}
[very thick, color=black,->,scale=1.8]
\node (n0) at (0,0) {000};
\node  [color=red] (n1) at (1.2,1) {\bf 001};
\node  [color=red] (n2) at (0,1) {\bf 010};
\node (n3) at (1.2,2) {011};
\node  [color=red] (n4) at (-1.2,1) {\bf 100};
\node (n5) at (0,2) {101};
\node (n6) at (-1.2,2) {110};
\node  [color=red]  (n7) at (0,3) {\bf 111};
\draw (n0) -- (n1);
\draw (n0) -- (n2);
\draw (n0) -- (n4);
\draw (n3) -- (n1);
\draw (n3) -- (n2);
\draw (n3) -- (n7);
\draw (n5) -- (n1);
\draw (n5) -- (n4);
\draw (n5) -- (n7);
\draw (n6) -- (n2);
\draw (n6) -- (n4);
\draw (n6) -- (n7);
\end{tikzpicture}
\caption{The fitness graph has 4 peaks. The  graph implies 3-way epistasis.}
\end{figure}
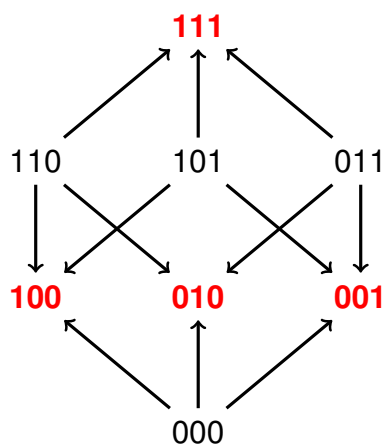

The topic of higher order epistasis and peaks was considered in \citet{cgg}, mainly for 3-locus systems.
There are 1862 fitness graphs for 3-locus systems, partitioned into 54 isomorphism classes.
20 out of the 54 graphs imply 3-way epistasis. The relation between higher order epistasis and 
peaks can be summarized as follows \citep{cgg}.
\begin{itemize}
\item[(i)]
The mean number of peaks for the graphs implying 3-order epistasis is 1.8;
\item[(ii)]
the mean number of peaks for the remaining graphs is 1.5,
\item[(iii)]
graphs with the maximal number of peaks imply 3-way epistasis.
\end{itemize}

However,  there are 193,270,310 fitness graphs for 4-locus systems (see the Supplementary Information)
and some theory is necessary for determining if a fitness graph implies $4$-way epistasis.
More generally, it is of interest to determine whether or not a given partial order 
for an $n$-locus system implies $n$-way epistasis. 
%A sufficient condition for partial orders is given in \cite[Proposition 7]{cgg}.
%Here we show that the same condition is necessary (see Theorem 2.3).
Here we develop the necessary theory, 
and use the results for a systematic investigation of higher 
order epistasis and peaks for 4-locus systems.

%Similar to the 3-locus case, we observe that graphs with 
%many peaks more frequently imply higher order epistasis. For graphs with 6 or more peaks,
%81 percent imply 4-way epistasis, whereas we estimate that 9 percent of all 4-locus fitness
%graphs imply 4-way epistasis. 

%Whether or not the partition property is necessary was  announced as an open problem
%at a conference " at the Max Planck institute in Leipzig (see also the acknowledgement).  

\section{Results}
As mentioned in the introduction, 3-way epistasis is measured by 
\[
u_{111} = w_{000} + w_{011} + w_{101} + w_{110}  - w_{001} - w_{010} - w_{100} -w_{111}.
\]
Notice that the coefficients for even genotypes are positive, and for odd genotypes negative,
where even and odd are as in the introduction.
For a given $n$-locus system it is useful to consider $k$-locus subsystems.
For instance, the genotypes with last coordinate zero
\[
000, 100, 010, 110,
\]
constitute a subsystem.
The conditional epistasis relevant for the 2-locus subsystem is defined by the equation
\[
w_{000}+w_{110}-w_{100}-w_{010}.
\]
There are in total six 2-locus subsystems (corresponding to the sides of the cube),
and each subsystem is defined by fixing one of the coordinates
at 0 or 1.

\begin{definition}
An $n$-locus system has $k$-way epistasis for $2 \leq  k \leq n$ if it has 
conditional $k$-way epistasis for at least one $k$-locus subsystem.
\end{definition}
Notice  that similar concepts in the  literature may disagree with this definition.
It is easy to verify that if a system has no $k$-way epistasis,
then it does not have $k+1$ epistasis. 
Fitness landscapes fall naturally into three categories:
\begin{itemize}
\item[(i)] additive landscapes;
\item[(ii)] landscapes with pairwise interactions, but not higher order epistasis,
\item[(iii)] landscapes with higher order epistasis.
\end{itemize}
If one wants to be more precise, one can say that a landscape has $k$-way epistasis
but not higher order interactions.

Some rank orders imply 3-way epistasis, in the sense that every fitness landscape
that satisfy the rank order has 3-way epistasis.
For instance, it is easily seen that
$u_{111}>0$ for the order 
\[
w_{000} >w_{011} > w_{101} > w_{110}  >w_{001} > w_{010}>w_{100} > w_{111},
\]
since the fitness for every single even genotype exceeds the fitness of the odd genotypes.
A complete characterization of rank orders that imply $n$-way epistasis is given in \cite{cgg}.
Here we are interested in understanding when partial orders, in particular fitness graphs,
imply $n$-way epistasis.
Fitness graph are of special interest since they capture important
aspects of the evolutionary dynamics. We illustrate  our method for the fitness graph in Figure 1. 
The 8 vertices of the graph can be partition into 4 pairs
\[
(000, 100),  \, (110,010),  \, (101, 001),  \, (011, 111)
\]
such that the even node in each pair has the lowest fitness.
It is easy to verify that $u_{111}<0$, i.e., the fitness graph implies 3-way epistasis.

More generally,  it is straight forward to verify 
that any fitness graph that satisfies a similar partition 
condition implies higher order epistasis  \citep{cgg}.
We will show that the partition property is a necessary condition for
the fitness graph to imply epistasis.
The proof relies on the following result for bipartite graphs.

\noindent
{\bf{Halls Marriage Theorem}} \citep{hall}.
Given a bipartite graph $G=(V, E)$ with bipartition $A \cup B $ $(V=A \cup B)$,
the graph $G$ has a matching of size $|A|$ if and only if for every $S \subseteq A$,
we have $|N(S)| \geq |S|$, where 
\[
N(S)= \{  b \in B, \text{ such that there exists } a \in S \text{ with } (a,b) \in E \}.
\].

Before we give a proof, 
we demonstrate the main idea of our argument in an example.

\begin{example}
Assume that 
\[ 
1111 \succ 1110   \text{ and }  1111 \succ 1101, 
\]
and that $\prec$ does not constrain 
$1110$  and $1101$ in any other way.
Then we can assign fitnesses such that 
\[
w_{1111}  - w_{1110}-w_{1101}
\]
is an arbitrarily large negative number. By assigning fitnesses to
the remaining  variables we can define a fitness landscape which respects $\prec$,
such that $u_{1111} <0$.
\end{example}

A  similar argument will be used in the proof.

\begin{theorem}
Let   $o_i$ denote the odd genotypes and $e_i$ the even genotypes.
A partial order $\prec$ of the genotypes with respect to fitness implies (positive) $n$-way epistasis if and only if
there exist a partition $\{  o_j, e_j  \}$ of all genotypes such that
\[
o_j  \prec  e_j
\] for all $j$.
\end{theorem}

\begin{proof}
It is straight forward to verify that the condition is sufficient \citep{cgg}.
Consider all  fitness landscapes that respect a partial order $\prec $ of the genotypes.
Let $V_o$ denote the genotypes with odd subindices and $V_e$ the genotypes with
even subindices. Consider the graph with vertices $V_e \cup V_o$, where $(o, e)$ is an edge if $o \prec e$.

Let $N$ be defined as above. Suppose that there exists a set of nodes $S_o \subseteq V_o$ 
such that $|N(S_o)| < |(S_o)| $. Then we can assign fitnesses to the set of nodes in $S_o \cup N( S_o )$, such
that  
\[
\sum_{ e_i \in N(S_o)}  w(e_i) - \sum_{o_j \in S_o} w(o_j)
\]
is an arbitrarily large  negative number. It follows that one can define
a fitness landscape which respects $\prec$ and does not have (positive) $n$-way epistasis.
Consequently, if the partial order $\prec$ implies $n$-way epistasis, then there does not exists any set $S_o$ as described.
By Hall's theorem, a matching exists and therefore the desired partition.
\end{proof}

Theorem 2.3 answers an open problem posed at a conference on interactions between algebra and the sciences
at the Max Planck institute in Leipzig, May 2017 (see also the acknowledgement).

The 193,270,310 fitness graphs for 4-locus systems can be partitioned into
511,863 isomorphism classes (see the Supplementary Information).
In order to analyze the relation between peaks and higher order
epistasis, we investigated all  graphs with 6 or more peaks.
The maximal number of peaks for a 4-locus system is 8 \citep{h}, see Figure 2.
\begin{figure}
\begin{tikzpicture}
[very thick, color=black,->, scale=1.2]
%[very thick,black,->, outer sep=1mm, scale 1.2]
\node (n0) at (0, -4) {0000};
\node  [color=red] (n1) at (4.0, -2.0) {\bf0001};
\node  [color=red] (n2) at (1.3333333333333335, -2.0) {\bf 0010};
\node (n3) at (5.0, 0) {0011};
\node [color=red] (n4) at (-1.333333333333333, -2.0) {\bf 0100};
\node (n5) at (3.0, 0) {0101};
\node (n6) at (1.0, 0) {0110};
\node  [color=red] (n7) at (4.0, 2.0) {\bf 0111};
\node  [color=red] (n8) at (-3.9999999999999996, -2.0) {\bf 1000};
\node (n9) at (-1.0, 0) {1001};
\node (n10) at (-3.0, 0) {1010};
\node  [color=red] (n11) at (1.3333333333333335, 2.0) {\bf 1011};
\node (n12) at (-5.0, 0) {1100};
\node  [color=red] (n13) at (-1.333333333333333, 2.0) {\bf 1101};
\node  [color=red] (n14) at (-3.9999999999999996, 2.0) {\bf 1110};
\node (n15) at (0, 4) {1111};
\draw (n0) -- (n1);
\draw (n0) -- (n2);
\draw (n0) -- (n4);
\draw (n0) -- (n8);
\draw (n3) -- (n1);
\draw (n3) -- (n2);
\draw (n3) -- (n7);
\draw (n3) -- (n11);
\draw (n5) -- (n1);
\draw (n5) -- (n4);
\draw (n5) -- (n7);
\draw (n5) -- (n13);
\draw (n6) -- (n2);
\draw (n6) -- (n4);
\draw (n6) -- (n7);
\draw (n6) -- (n14);
\draw (n9) -- (n1);
\draw (n9) -- (n8);
\draw (n9) -- (n11);
\draw (n9) -- (n13);
\draw (n10) -- (n2);
\draw (n10) -- (n8);
\draw (n10) -- (n11);
\draw (n10) -- (n14);
\draw (n12) -- (n4);
\draw (n12) -- (n8);
\draw (n12) -- (n13);
\draw (n12) -- (n14);
\draw (n15) -- (n7);
\draw (n15) -- (n11);
\draw (n15) -- (n13);
\draw (n15) -- (n14);
\end{tikzpicture}
\caption{The fitness graph has 8 peaks. The graph implies $4$-way epistasis. }
\end{figure}
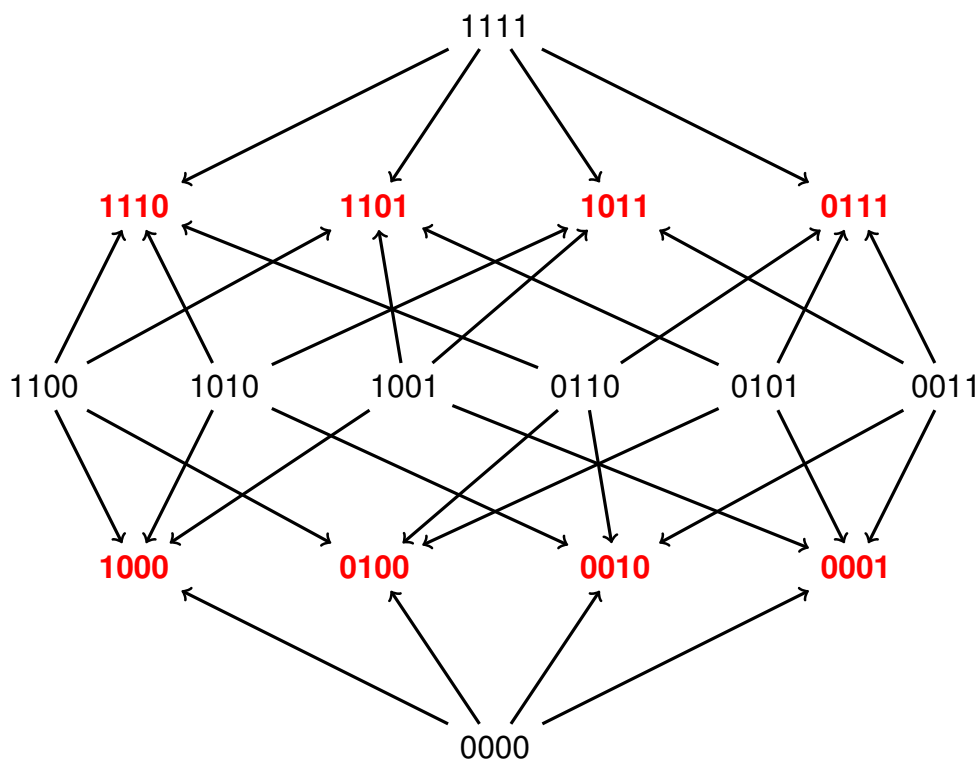

Table 1 summarizes the interactions for 4-locus systems with 6-8 peaks.
There are in total 64 graphs in the category,
52 out of them imply $4$-way epistasis.
The sole graph that does not imply higher order epistasis has 6 peaks.
\begin{table}
\begin{tabular}{|c|c|c | c | c |}
$ \#$ Peaks & $\#$ Graphs & 2-way  & $3$-way & 4-way\\
\hline \hline
8  & 1 &  0 & 0 & 1 \\ 
7  & 4 & 0 & 1 & 3\\
6  & 59 & 1  & 10 & 48\\
\hline
6--8 & 64 & 1& 11 & 52 \\
\hline
\end{tabular}
\caption{The table shows how gene interactions and the number of peaks relate for 4-locus systems. 
The rightmost column shows the number of graphs that imply  4-way epistasis. The
adjacent column shows the number of graphs that imply 3-way espistasis, but
not higher order epistasis. The middle columns shows the number of graphs that are compatible
with pairwise epistasis only.}
\end{table}

Table 2 summarizes the relation between $4$-way epistasis 
and peaks. The information for graphs with
$6-8$ peaks  is exact. The remaining statistics 
depend on a sample of 200 randomly chosen graphs 
among the $511, 863$  isomorphism classes.

\begin{table}
\begin{tabular}{|c| c |}
$ \#$ Peaks & 4-way interactions \\
\hline \hline 
$6-8$  & 81 \% \\ 
5  & 50  \% \\ 
4  & 15  \%\\
3  & 10  \% \\
2  & 10  \%\\
1  & 5  \% \\
\hline
All graphs & 9 \% \\
\hline 
\end{tabular}
\caption{The table shows how the proportion of graphs that imply 4-way epistasis depends on the number of peaks.}
\end{table}

The maximal number of peaks for an $n$-locus system is $2^{n-1}$ \citep{h}.
However, in the absence of higher order epistasis, the maximal number of peaks is 3 (rather than 4 in general) for
3-locus systems, and 6 (rather than 8 in general) for 4-locus systems.
A landscape compatible with the graph in Figure 2 is defined as
\[
  w_g = 1+   0.1 d  - { d  \choose 2}  0.06   \quad  \text{ where}  \quad   d= \sum g_i  .
\]
Notice that the fitness of each genotype is determined by the number or $1$'s. In particular,
\[
w_{0000}=1, \, \,
w_{1000}=1.1, \, \,
w_{1100}=1.14, \, \,
w_{1110}=1.12,  \, \,
w_{1111}=1.04.
\]

A similar construction works for any number of loci (see the Supplementary Information).
One concludes that the  maximal number of peaks for landscapes with  only pairwise interactions is
 \[
 \geq  { n  \choose \,   \lfloor  \frac{n}{2} \rfloor  \, } .
  \]
The bound is exact for $n=2, 3, 4$. 
Whether or not the bound is exact in general is an open problem.
Asymptotically,
 \[ 
 { n  \choose \,   \lfloor{  \frac{n}{2} \rfloor}  \, }   \mapsto \frac{4^n}{ \sqrt{\pi n} }
 \]
(see also the Supplementary Information).

One can ask if the landscapes with extreme peak density
are likely to occur in nature, or rather should be considered
theoretical constructions. In fact, the unique fitness graph 
with 6 peaks that is compatible with pairwise interactions 
is biologically meaningful (Figure 3).  The graph is compatible with
stabilizing selection. Stabilizing selection occurs for traits where 
an intermediate value is preferable as compared to the
extremes in the population.
Notice that the six genotypes
\[
1100, 1010, 1001, 0110, 0101, 0011,
\]
have the same number of $0$'s and $1$'s,
and any mutation that makes the distribution more equal is
beneficial.

Moreover, the fitness graph with eight peaks (Figure 2)
is compatible with egg box models \citep{bmh}.
Egg box models agree with some empirical  data sets,
although such examples are probably not common.

\bigskip

\section{Discussion}
We have studied the relation between higher order epistasis
and peaks in fitness landscapes.
Fitness graphs (directed acyclic hypercube graphs) reveal all possible
constellations of peaks in the landscapes.
We provide a characterization of fitness graphs 
that imply higher order epistasis. A similar result holds for partial orders.
Our proof depends on Hall's marriage theorem for bipartite graphs. 
The theoretical results are of independent interest. 

We applied the characterization extensively in  a systematic study of fitness graphs
for 4-locus biallelic systems, in total 193,270,310 fitness graphs.
The graphs can be partitioned into 511, 863 isomorphism classes.
We identified all graphs with 6 or more peaks and found that
81 percent of them imply 4-way epistasis. In contrast, 9 percent of all fitness graphs 
imply 4-way epistasis. Our result agrees well with 
properties of 3-locus systems \citep{cgg}.

Fitness graphs allow for analyzing fitness landscapes in general, 
rather than focusing on a particular parametric model,
such as the NK model or the Mount Fuji model \citep[e.g.]{dk}.
Results on fitness graphs are in a sense universal, although one should keep in mind
that some graphs occur frequently in nature, whereas other graphs 
are rare. 

Fitness graphs have been applied for relating global properties,
such as the number of peaks and mutational trajectories, and local
properties of fitness landscapes \cite{cgb}.
The capacity of fitness graphs to detect higher order epistasis
was demonstrated in \cite{wdo}. The authors observed that one can deduce 3-way epistasis from
particular graphs. A complete description of fitness graphs that imply 3-way epistasis was provided in \cite{cgg}.
For related results, see also \citep{ptk,wwc}.
The graphs have potential for other aspects of higher order epistasis as well,
including the relation between higher order epistasis and (accessible) mutational trajectories.

One would expect that higher order epistasis implies that adaptation is more difficult
according to almost any measure. In any case, more systematic studies would be valuable.
%Although fitness graphs reveal mutational trajectories,
%they are of no help for computing probabilities of the trajectories.
An obvious limitation of fitness graphs is that they 
encode pairwise fitness comparisons only, 
and not for instance the curvature of fitness landscapes.
Other methods are necessary as well for analyzing higher order epistasis.

The utility of fitness landscapes have been discussed  in recent years \citep[e.g.][]{hartl}
and sometimes questioned altogether \citep[e.g.][]{kaplan}.
Fitness graphs and other efficient representations of  fitness landscapes
may ameliorate some of the issues addressed. 
In any case, fitness graphs are intuitive and at the same time well suited for
developing theory.

\begin{figure}
\begin{tikzpicture}
[very thick,black,->,outer sep=1mm, scale=1.2]
\node (n0) at (0, -4.0) {0000};
\node (n1) at (4.0, -2.0) {0001};
\node (n2) at (1.3333333333333335, -2.0) {0010};
\node [color=red] (n3) at (5.0, 0) {\bf 0011};
\node (n4) at (-1.333333333333333, -2.0) {0100};
\node  [color=red] (n5) at (3.0, 0) {\bf 0101};
\node [color=red] (n6) at (1.0, 0) {\bf 0110};
\node (n7) at (4.0, 2.0) {0111};
\node (n8) at (-3.9999999999999996, -2.0) {1000};
\node [color=red] (n9) at (-1.0, 0) {\bf 1001};
\node [color=red] (n10) at (-3.0, 0) {\bf1010};
\node (n11) at (1.3333333333333335, 2.0) {1011};
\node [color=red](n12) at (-5.0, 0) {\bf 1100};
\node (n13) at (-1.333333333333333, 2.0) {1101};
\node (n14) at (-3.9999999999999996, 2.0) {1110};
\node (n15) at (0, 4.0) {1111};
\draw (n0) -- (n1);
\draw (n0) -- (n2);
\draw (n0) -- (n4);
\draw (n0) -- (n8);
\draw (n1) -- (n3);
\draw (n1) -- (n5);
\draw (n1) -- (n9);
\draw (n2) -- (n3);
\draw (n2) -- (n6);
\draw (n2) -- (n10);
\draw (n4) -- (n5);
\draw (n4) -- (n6);
\draw (n4) -- (n12);
\draw (n7) -- (n3);
\draw (n7) -- (n5);
\draw (n7) -- (n6);
\draw (n8) -- (n9);
\draw (n8) -- (n10);
\draw (n8) -- (n12);
\draw (n11) -- (n3);
\draw (n11) -- (n9);
\draw (n11) -- (n10);
\draw (n13) -- (n5);
\draw (n13) -- (n9);
\draw (n13) -- (n12);
\draw (n14) -- (n6);
\draw (n14) -- (n10);
\draw (n14) -- (n12);
\draw (n15) -- (n7);
\draw (n15) -- (n11);
\draw (n15) -- (n13);
\draw (n15) -- (n14);
\end{tikzpicture}
\caption{The fitness graph has 6 peaks. The graph is compatible with pairwise interactions only.}
\end{figure}
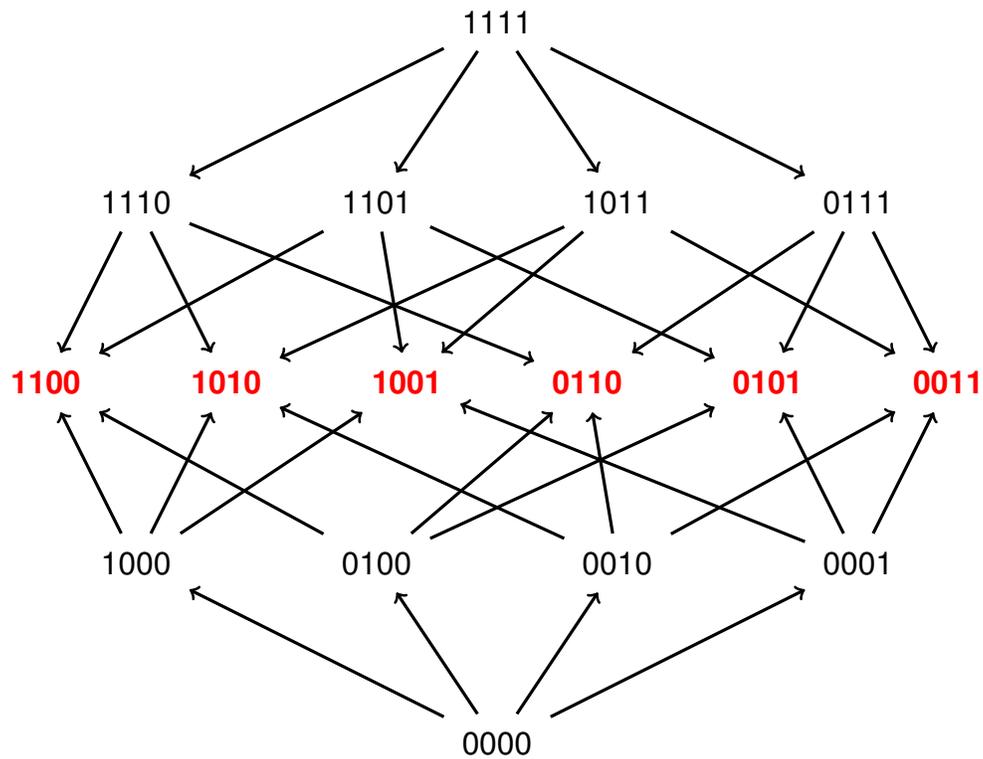

\begin{figure}
\begin{tikzpicture}
[very thick, color=black,->, scale=1.2]
\node (n0) at (0, -4.0) {0000};
\node [color=red] (n1) at (4.0, -2.0) {\bf 0001};
\node [color=red](n2) at (1.3333333333333335, -2.0) {\bf0010};
\node (n3) at (5.0, 0) {0011};
\node [color=red](n4) at (-1.333333333333333, -2.0) {\bf 0100};
\node (n5) at (3.0, 0) {0101};
\node (n6) at (1.0, 0) {0110};
\node (n7) at (4.0, 2.0) {0111};
\node [color=red] (n8) at (-3.9999999999999996, -2.0) {\bf 1000};
\node (n9) at (-1.0, 0) {1001};
\node (n10) at (-3.0, 0) {1010};
\node  [color=red](n11) at (1.3333333333333335, 2.0) {\bf 1011};
\node (n12) at (-5.0, 0) {1100};
\node  [color=red] (n13) at (-1.333333333333333, 2.0) {\bf 1101};
\node  [color=red] (n14) at (-3.9999999999999996, 2.0) {\bf 1110};
\node (n15) at (0, 4.0) {1111};
\draw (n0) -- (n1);
\draw (n0) -- (n2);
\draw (n0) -- (n4);
\draw (n0) -- (n8);
\draw (n3) -- (n1);
\draw (n3) -- (n2);
\draw (n3) -- (n11);
\draw (n5) -- (n1);
\draw (n5) -- (n4);
\draw (n5) -- (n13);
\draw (n6) -- (n2);
\draw (n6) -- (n4);
\draw (n6) -- (n14);
\draw (n7) -- (n3);
\draw (n7) -- (n5);
\draw (n7) -- (n6);
\draw (n7) -- (n15);
\draw (n9) -- (n1);
\draw (n9) -- (n8);
\draw (n9) -- (n11);
\draw (n9) -- (n13);
\draw (n10) -- (n2);
\draw (n10) -- (n8);
\draw (n10) -- (n11);
\draw (n10) -- (n14);
\draw (n12) -- (n4);
\draw (n12) -- (n8);
\draw (n12) -- (n13);
\draw (n12) -- (n14);
\draw (n15) -- (n11);
\draw (n15) -- (n13);
\draw (n15) -- (n14);
\end{tikzpicture}
\caption{The fitness graph has 7 peaks. The  graph does not imply  4-way epistasis.
However, it implies 3-way epistasis.}
\end{figure}
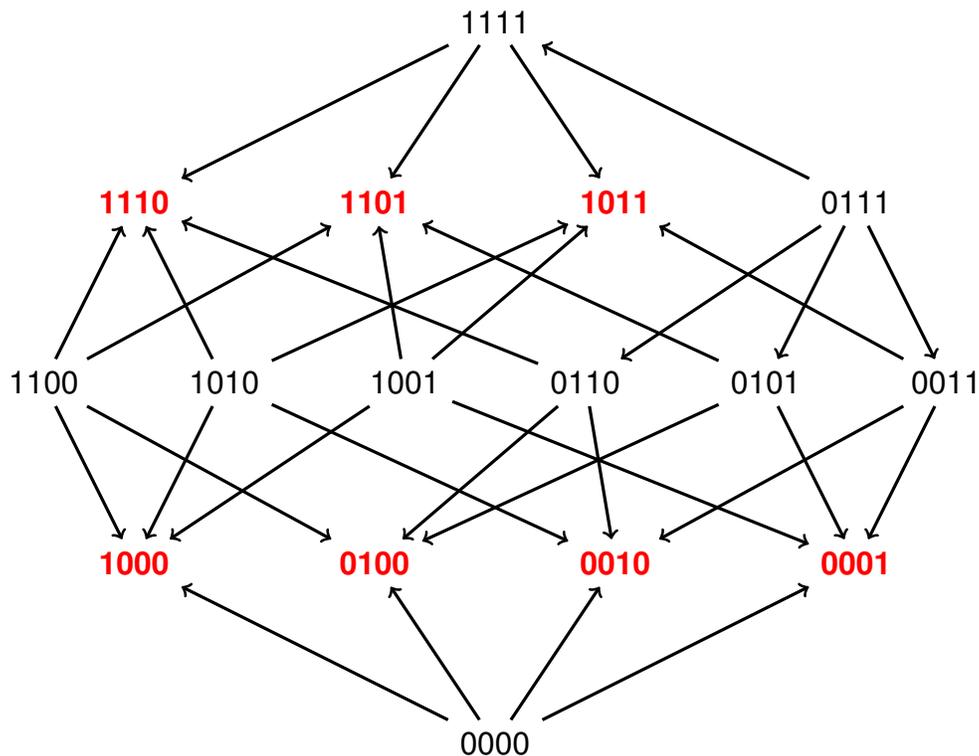

\bigskip
\noindent
{\bf{Acknowledgements.}}
Devin Greene produced all  code for 
the project. The code made it possible to
identify  isomorphism classes of graphs, 
compute the number of peaks and automatically 
produce pictures.

The partition result, here Theorem 2.3,
was proved independently by another group,
and will  also appear in work
 by Caitlin Lienkaemper and co-authors.

\section{Supplementary information}

\subsection*{4-locus fitness graphs}
There are 193,270,310 fitness graphs for 4-locus systems.
This can be verified by considering the chromatic polynomial for
4-cube graphs.
The chromatic polynomial is
\[
\begin{aligned}
& (x-1)x (x^{14}-31 x^{13}+465 x^{12}-4471 x^{11}+30793 x^{10}-160807 x^9 +657229 x^8 \\
&-2137667 x^7  +5564285 x^6-11536667 x^5 +18740317 x^4 -23081607 x^3 \\
&+ 20308039 x^2-11372201x+3040575).
\end{aligned}
\]
Evaluation of the polynomial in $-1$ gives  193,270,310.
By Corollary 1.3 in  \citet{s},
there are 193,270,310 fitness graphs for 4-locus systems.
The graphs can be partitioned into 511,863 isomorphism classes.

The argument for finding the isomorphism classes is similar to
the 3-locus case \citep{cgg}.
We generated a list of 511,863 fitness graphs, one for each isomorphism class,
and identified all graphs with 6 or more peaks. We sampled 
graphs with fewer peaks from the list (summarized in Table 2).

\subsection*{Fitness landscapes with no higher order epistasis}
Consider fitness landscapes of the following type:
\[
  w_g = 1+   k a  - { k  \choose 2}  e   \quad  \text{ where}  \quad   k= \sum g_i  
\]
where $a$ and $e$ are positive numbers. 
Moreover, assume that $a$ and $e$ are chosen so that $w_g>0$ for all $g$.
For $n=4$ we obtain the fitness landscapes discussed in the main text
if $a=0.1$ and $e=0.06$. All  genotypes with $\sum g_i =2$ are peaks. 

The same construction works for any $n$. It is straight forward to verify that $a$ and $e$
can be chosen such that $w_g$ is maximal if $k=\lfloor  \frac{n}{2} \rfloor$.
Such graphs have 
$
{ n  \choose \,   \lfloor  \frac{n}{2} \rfloor  \, } .
$
peaks.
For the category of fitness landscapes with no higher order epistasis,
one concludes that the maximal number of peaks is
\[
\geq { n  \choose \,   \lfloor  \frac{n}{2} \rfloor  \,}.
\]
The bound is exact for $n \leq 4$, as is clear from Table 1 for  $n=4$,
and from \citet{cgg} for $n=3$.
For asymptotic behavior of this expression as described in the main text, see also
\url{http://oeis.org/A000984} and references therein.

%\[
%\frac{2^n} {\sqrt{\frac{n \pi}{2}}}
%\]
%The peak density approaches $\sqrt{ \frac{2}{n \pi} }$.
%Conjecture the maximal peak in the absence of epistasis is of the order  $\frac{ 1}{\sqrt n} }$.
%Online Encyclopedia of Integer Sequences is:
%A000984 - OEIS
%https://www.quora.com/What-is-the-asymptotic-value-of-n-choose-n-2-Is-it-around-n2-n

\end{document}